\documentclass[a4paper,UKenglish,cleveref, autoref, thm-restate]{lipics-v2021}

\newcommand{\itnot}{\textit{not }}
\newcommand{\grp}[1]{\left(#1\right)}
\newcommand{\set}[1]{\left\{#1\right\}}
\newcommand{\abs}[1]{\left|#1\right|}
\newcommand{\sub}{\subseteq}

\DeclareMathOperator{\OPT}{OPT}

\newcommand{\eps}{\epsilon}

\DeclareMathOperator{\lcs}{LCS}
\DeclareMathOperator{\ed}{ED}
\DeclareMathOperator{\bm}{BestMatch}
\DeclareMathOperator{\match}{Match}
\DeclareMathOperator{\greed}{Greedy}
\DeclareMathOperator{\approxed}{ApproxED}

\title{Improved Approximation for Longest Common Subsequence over Small Alphabets}

\author{Shyan Akmal}{MIT EECS and CSAIL, USA  }{naysh@mit.edu}{https://orcid.org/0000-0002-7266-2041}{Supported by NSF Grant CCF-1909429.}

\author{Virginia Vassilevska Williams}{MIT EECS and CSAIL, USA}{virgi@mit.edu}{}{Supported by an NSF CAREER Award, NSF Grant CCF-1909429, a BSF Grant BSF:2012338, a Google Research Fellowship and a Sloan Research Fellowship.}

\authorrunning{S.\,S. Akmal and V.\,V. Williams} 

\Copyright{Shyan S. Akmal and Virginia Vassilevska Williams}

\ccsdesc[500]{Theory of computation~Design and analysis of algorithms}

\keywords{approximation algorithms, longest common subsequence, subquadratic} 

\category{Track A: Algorithms, Complexity and Games} 

\relatedversion{} 

\acknowledgements{We thank Jenny Kaufmann for suggesting several helpful revisions.}

\nolinenumbers 

\EventEditors{John Q. Open and Joan R. Access}
\EventNoEds{2}
\EventLongTitle{42nd Conference on Very Important Topics (CVIT 2016)}
\EventShortTitle{CVIT 2016}
\EventAcronym{CVIT}
\EventYear{2016}
\EventDate{December 24--27, 2016}
\EventLocation{Little Whinging, United Kingdom}
\EventLogo{}
\SeriesVolume{42}
\ArticleNo{23}

\begin{document}

\maketitle

\begin{abstract}
This paper investigates the approximability of the Longest Common Subsequence (LCS) problem. 
The fastest algorithm for solving the LCS problem exactly runs in essentially quadratic time in the length of the input, and it is known that under the Strong Exponential Time Hypothesis the quadratic running time cannot be beaten. 
There are no such limitations for the approximate computation of the LCS however, except in some limited scenarios. 
There is also a scarcity of approximation algorithms.
When the two given strings are over an alphabet of size $k$, returning the subsequence formed by the most frequent symbol occurring in both strings achieves a $1/k$ approximation for the LCS. It is an open problem whether a better than $1/k$ approximation can be achieved in truly subquadratic time ($O(n^{2-\delta})$ time for constant $\delta>0$). 

A recent result [Rubinstein and Song SODA'2020] showed that a $1/2+\epsilon$ approximation for the LCS over a binary alphabet is possible in truly subquadratic time, provided the input strings have the same length. 
In this paper we show that if a $1/2+\epsilon$ approximation (for $\epsilon>0$) is achievable for binary LCS in truly subquadratic time when the input strings can be unequal, then for every constant $k$, there is a truly subquadratic time algorithm that achieves a $1/k+\delta$ approximation for $k$-ary alphabet LCS for some $\delta>0$. 
Thus the binary case is the hardest. 
We also show that for every constant $k$, if one is given two strings of \emph{equal} length over a $k$-ary alphabet, one can obtain a $1/k+\epsilon$ approximation for some constant $\epsilon>0$ in truly subquadratic time, thus extending the Rubinstein and Song result to all alphabets of constant size.
\end{abstract}

\section{Introduction}

In a large variety of applications, from spell-checkers to DNA sequence alignment, one seeks to compute how similar two given sequences of letters are. Arguably the most popular measures of sequence similarity are the 
 \emph{longest common subsequence (LCS)} and the  \emph{edit distance}. 
The LCS of two given sequences $A$ and $B$ 
 (as the name suggests) measures the maximum length of a sequence whose symbols appear in both $A$ and $B$ in the same order.
The edit distance, on the other hand, measures how far apart two strings are by counting the minimum number of insertions, deletions and substitutions of characters that must be performed on one string to transform it into the other. These two measures are related: the complement of the LCS is the version of edit distance in which one minimizes only the number of insertions and deletions (in fact, both versions of edit distance are the same up to a factor of $2$, so with respect to constant factor approximation algorithms they are equivalent).
Both the LCS and the edit distance of two length $n$ strings can be computed in $O(n^2)$ time using a classic dynamic programming approach. 
The fastest algorithm for both problems is the $O(n^2/(\log n)^2)$ time algorithm of Masek and Paterson \cite{MasekP80}.

Hardness results from fine-grained complexity have shown that truly subquadratic time algorithms (those running in time $O(n^{2-\delta})$ for some constant $\delta>0$) for LCS and edit distance cannot exist under the Strong Exponential Time Hypothesis \cite{AbboudBW15, BackursI18, Bringmann2018} and other even more believable hypotheses \cite{logshaving-hard,ChenGLRR19}. 
Consequently, much of the recent interest around LCS and edit distance has concerned approximation algorithms for the problems.
A long chain of progress on edit distance approximation (e.g. \cite{lcs-root-approx,BatuES06,lcs-root-approx,Bar-YossefJKK04,AndoniO12,AndoniKO10,BoroujeniEGHS18}) has culminated in the breakthrough constant-factor approximation algorithm of \cite{ChakrabortyDGKS18} running in truly subquadratic time. Several improvements followed this breakthrough including \cite{simpler-ED, kouckysaks20, brakensiekrubin20}, with the most recent result being a constant factor approximation algorithm running in near linear time \cite{andoni-nearlinear-constant}.

In contrast,
much less is known about how well LCS can be approximated in truly subquadratic time. For inputs over non-constant size alphabets, some fine-grained hardness results \cite{AbboudB17,LCS-circuitlbs,ChenGLRR19} and nontrivial super-constant approximations \cite{linear-LCS} are known.
When the input strings come from a fixed alphabet of constant size $k$, there is a trivial $1/k$-approximation algorithm that returns in linear time the longest common \emph{unary} subsequence of two inputs.
Despite the simplicity of this algorithm, until recently no better constant-factor approximation algorithm was known for \emph{any} constant size alphabet. 
There are also no existing hardness results that rule out better approximation algorithms.

Recently, for the case of {\em binary} strings  of {\em equal} length, Rubinstein and Song \cite{approxED-LCS} were able to improve upon the simple $1/2$-approximation algorithm described above, obtaining for some constants $\epsilon,\delta>0$ an $O(n^{2-\delta})$ time algorithm that returns a $(1/2+\epsilon)$-approximation. 
It is still open, however, whether one can obtain such an algorithm for binary strings of unequal length, and whether one can extend 
the result to achieve a truly subquadratic time algorithm achieving a better than $k$-approximation for strings over an alphabet of size $k$, for every constant $k$.

\paragraph*{Our results.}
In this paper we present two results. 
The first result shows that if one can obtain a truly subquadratic time, better than $2$-approximation algorithm for the LCS of binary strings of possibly unequal lengths, then one can use this algorithm to obtain a truly subquadratic time, better than $k$-approximation algorithm for the LCS of strings over a $k$-ary alphabet, for any constant $k$.

\begin{theorem}\label{reduction-intro}
For any fixed integer $k\geq 2$ there is an $O(n)$ time algorithm that given an instance of LCS for strings of length at most $n$ over an alphabet of size $k$, reduces it to $O(k^2)$ instances of LCS over {\em binary} strings of length at most $n$, so that $(1/2+\epsilon)$-approximate solutions (for $\epsilon>0$) for these LCS instances can be translated in $O(n)$ time into a $(1/k+\epsilon_k)$ approximation of the $k$-ary alphabet LCS instance, where $\epsilon_k>0$ is a constant only depending on $\epsilon$ and $k$.
\end{theorem}

In other words, in order to beat the longstanding $1/k$-approximation algorithm for LCS, one merely needs to obtain a better than $1/2$ approximation for the binary case, i.e. to extend the Rubinstein and Song result to strings of possibly unequal length.

Our second result generalizes the Rubinstein-Song result by proving that one can beat the simple $1/k$-approximation algorithm for every constant $k$, as long as the input strings have equal length.

      \begin{theorem}
            \label{thm:equal-length-intro}
            Given two strings $A$ and $B$ of length $n$ over an arbitrary alphabet of size $k$, there exist positive constants $\epsilon$ and $\delta$ such that we can compute a $1/k + \eps$ approximation for the longest common subsequence of $A$ and $B$ in $O(n^{2-\delta})$ time.
        \end{theorem}
In fact, our algorithm can actually $(1/k+\eps)$-approximate the LCS in near-linear time.
Note that our result applies only to strings with equal input lengths.
The relevance of this restriction is discussed in \Cref{sec:input-length}, which also contains the proof of Theorem~\ref{reduction-intro}.
We present the proof of \Cref{thm:equal-length-intro} in \Cref{sec:alph-size}.

\paragraph*{Preliminaries.}
We write approximation ratios as constants less than $1$,
so that for example a $1/2 + \epsilon$ approximation algorithm for the LCS of $A$ and $B$ is an algorithm that returns a common subsequence of $A$ and $B$ with length at least $(1/2 + \epsilon)\cdot \lcs(A,B)$. 

When we discuss edit distance in the rest of the paper we mean the version of edit distance that does not allow symbol substitutions but only measures the number of insertions and deletions. For constant factor approximation algorithms, this version of the problem is equivalent to the original.

\section{Reduction to Binary Alphabets \& Input Length Conditions}
\label{sec:input-length}

We begin by showing how to reduce nontrivial constant factor approximations of LCS over large alphabets to better than $1/2$ approximations of LCS over binary alphabets.
Although we do not directly apply this reduction in our proof of \Cref{thm:equal-length-intro}, the reduction is elegant and motivates the approach we end up using. Moreover, the reduction works even for strings of non-equal lengths, thus showing that one merely needs to extend the Rubinstein-Song result to non-equal length strings in order to truly improve upon the trivial alphabet-size approximation algorithm.
We will need the following definition.

\begin{definition}[Restrictions]
Given an alphabet $\Sigma$,
we call a subset $\Sigma'\sub\Sigma$ a \emph{subalphabet}.
Given a string $A$ from alphabet $\Sigma$ the \emph{restriction} of $A$ to a subalphabet $\Sigma'$ is the 
    maximum subsequence of $A$ whose characters are all in $\Sigma'$.
\end{definition}

\begin{theorem}
    \label{thm:alph-reduction}
    Fix integers $s$ and $\ell$ with $s > \ell \ge 2$.
    Suppose that there is a $T(n,\ell)$ time algorithm that achieves an $1/(\ell-\eps)$-approximation of the LCS of two strings of length at most $n$ from an alphabet of size $\ell$. Then, there is also a $O((n+T(n,\ell))\binom s\ell)$ time algorithm that achieves an $1/(s(1-\eps/\ell))$-approximation of the LCS of two strings of length at most $n$ from an alphabet of size $s$.
\end{theorem}

\begin{proof}
We will show how to reduce the LCS for two strings of length at most $n$ over a $s$-ary alphabet, to the LCS for two strings of length at most $n$ over 
an $\ell$-ary alphabet for any $\ell<s$. 
The reduction runs in $O(n{s\choose \ell})$ time and produces ${s \choose \ell}$ instances of $\ell$-ary alphabet LCS.

Let $A$ and $B$ two strings of length at most $n$ over an alphabet $\Sigma$ of size $s$.
Let $C$ be the longest common subsequence of $A$ and $B$ (we do not know $C$). 
For the sake of argument, sort the alphabet symbols according to their number of occurrences in $C$. 

Let $x$ be the collection of the $\ell$ most frequent alphabet symbols in $C$. 
Let $C_x$ be the subsequence of $C$ obtained by restricting $C$ to the subalphabet of $\Sigma$ that contains the symbols of $x$. 
Since $x$ has the $\ell$ most frequent symbols in $C$, $C_x$ contains at least an $\ell/s$ fraction of $C$.
 
Now, let us describe our algorithm. 
Given $A$ and $B$, we consider all subsets of the alphabet consisting of precisely $\ell$ symbols
(one of these subsets will be $x$.)
For each such collection $y$, consider the sub-instance of the LCS instance restricted to the symbols of $y$. 
Let $\OPT(y)$ be the optimal LCS for this instance.

We know that
$|\OPT(x)|\geq |C_x|\geq (\ell/s) |C|$. 
So when we consider $y=x$, if we can efficiently obtain an $1/(\ell-\eps)$  approximation for $\OPT(x)$, we will get a common subsequence of $A$ and $B$ of length at least
\[\frac{|\OPT(x)|}{\ell-\eps} \geq \frac{|C|}{s(1-\eps/\ell)},\] 
and thus yields the desired approximation for the LCS of $A$ and $B$.
The running time is multiplied by $\binom s\ell$ which is a constant, as long as $s$ is.
\end{proof}

By setting $\eps=\ell\delta s/(1+\delta s)$ in the above Theorem statement, we obtain a linear time reduction from obtaining a $1/s+\delta$-approximation for $s$-ary strings to a $1/\ell + (\delta s)/\ell$-approximation for $\ell$-ary strings.
We obtain \Cref{reduction-intro} as a corollary:
\begin{corollary}
    \label{corr:bin-reduction}
    Fix an integer $s\ge 3$ and a constant $\delta>0$.
    The problem of obtaining a $1/s+\delta$
    approximation for the LCS of two strings from an alphabet of size $s$ can be reduced in linear time to the problem of obtaining a $1/2 + (\delta s)/2$ approximation for the LCS of two strings binary strings (i.e. strings from an alphabet of size two).
\end{corollary}

\begin{proof}
    This follows from \Cref{thm:alph-reduction} by taking $\ell = 2$. The reduction has $O(s^2)$ overhead.
\end{proof}

The reason we cannot prove \Cref{thm:equal-length-intro} by combining \Cref{corr:bin-reduction} with the result of \cite{approxED-LCS} is that the latter gives a better than $1/2$ approximation for strings from an alphabet of size $2$ only when the input strings have \emph{equal length}.
Note that in the reduction from the proof of \Cref{thm:alph-reduction}, the subsequences obtained from restrictions to subalphabets may be of different lengths even if the original strings have equal lengths.

Although at first it may seem that extending the Rubinstein-Song result to strings of differing length should not be too hard, generalizing the result does not appear straightforward.
In the following we will discuss why this not simple.
The main hurdles come from the lemma and algorithm used in \cite{approxED-LCS} stated below.

    \begin{lemma}[LCS and Edit Distance Connection]
        \label{lm:lcs-ed}
        For any strings $X$ and $Y$ of length $n$ and $m$ respectively,
        we have     
        \[2\cdot\lcs(X,Y) + \ed(X,Y) = n + m.\]
        
    \end{lemma}

    For the sake of completeness, we include a proof of the above lemma.
\begin{proof}
        
        Consider an optimal alignment between $X$ and $Y$, which matches the maximum possible number of characters of $X$ with identical characters of $Y$ while respecting the order in which the characters appear in each string.
        The characters that are matched in $X$ and $Y$ correspond to a longest common subsequence.
        This is because if there were a longer common subsequence, we could get a larger alignment by matching the characters of the subsequence in $X$ and $Y$, but this would contradict the optimality of $X$ and $Y$.
        
        Similarly, the unmatched characters correspond to a minimum set of symbols that need to be deleted from $X$ and inserted from $Y$ to turn $X$ into $Y$.
        If there were a smaller edit distance computation, then all the characters which were not deleted or inserted could be paired up to form a larger alignment, again contradicting optimality.
        
        Thus there are exactly $2\cdot \lcs(X,Y)$ characters paired up in the alignment and $\ed(X,Y)$ unmatched characters.
        These encompass all the characters in $X$ and $Y$, and thus account for $n+m$ symbols.
    \end{proof}

    \begin{definition}[Approximating LCS through Edit Distance]
    \label{def:approxed}
    Given strings $A$ and $B$ of length $n$, the algorithm $\approxed(A,B)$ approximates the edit-distance between $A$ and $B$ and then returns the lower bound on the LCS implied by this.
    More precisely, the algorithm computes an approximate edit distance $\widetilde{\ed}(A,B)$ and then returns
        \begin{equation}
        \label{eq:ed-subtract}
        n - \frac 12 \cdot \widetilde{\ed}(A,B).
        \end{equation}
    \end{definition}
    
    As stated $\approxed$ can use any edit distance approximation $\widetilde{\ed}$ as a black box.
    For concreteness, we will take $\widetilde{\ed}$
    to be the edit distance algorithm from \cite{simpler-ED} which can achieve an approximation ratio of $c$ for any constant $c > 3$ and runs in truly subquadratic time.
    
    In \cite{approxED-LCS}, the authors use the $\approxed$ algorithm to handle the case of binary strings with large LCS.
    In this case, they notice that if the LCS is large then the edit distance must be small.
    Then constant factor approximations to edit distance will give good lower bounds on the LCS because in the calculation from \Cref{eq:ed-subtract} we are subtracting off a small quantity from the maximum possible LCS value of $n$.
    
    However, if we tried extending this algorithm to the case where the inputs $X$ and $Y$ have lengths $n$ and $m$ with $n = 100m$ (for example) by using the identity from \Cref{lm:lcs-ed}, then even when the LCS is large (say of length $(1-\epsilon)m$ for some small positive $\epsilon$) the edit distance will still be very large compared to the length of the smaller string (at least $(99 + \epsilon)m$).
    So even a $3$-approximation to edit-distance would incur massive error when trying to approximate LCS by computing  
        \[\frac 12\cdot \grp{n + m - \widetilde{\ed}(X,Y)}\]
    and the result would not give any nontrivial lower bound for the LCS.
    In other words, when the input strings have very different lengths it is not clear how to use approximate edit distance in general to obtain good approximations for LCS.
    This is essentially why the algorithm from \cite{approxED-LCS} and \Cref{thm:equal-length-intro} both require equal length inputs.

\section{Extending Alphabet Size}
\label{sec:alph-size}
This section proves \Cref{thm:equal-length} which restates \Cref{thm:equal-length-intro} from the introduction slightly.

 \begin{theorem}
     \label{thm:equal-length}
Given two strings $A$ and $B$ of length $n$ over an arbitrary alphabet $\Sigma$ of size $s$, there exist a positive constant $\epsilon$ such that we can compute a $1/s + \eps$ approximation for the longest common subsequence $\lcs(A,B)$ of $A$ and $B$ in truly subquadratic time.
\end{theorem}

    Throughout the rest of this section, we assume $A$ and $B$ refer to strings satisfying the conditions of \Cref{thm:equal-length} and that $s\ge 3$.
    We begin by establishing lemmas corresponding to easy instances of the problem.
    The following definition is useful for identifying these easy cases.
    
    \begin{definition}[Balanced Strings]
    Given a string $A$ of length $n$ from an alphabet of size $s$ and a parameter $\rho > 0$,
    we say a string is \emph{$\rho$-balanced} 
    if all its character frequencies are within $\rho n$ of $n/s$.
    
    \end{definition}

        \begin{lemma}[Balanced Inputs, adapted from Lemma 3.2 of \cite{approxED-LCS}]
            \label{lm:balanced-string}
            For all sufficiently small $\rho > 0$, if either $A$ or $B$ is $\rho$-balanced,
            we can $\grp{1/s+\gamma}$-approximate $\lcs(A,B)$ in truly subquadratic time, where $\gamma$ is some positive constant depending on $\rho$.
        \end{lemma}
            \begin{proof}
                
                We reduce the problem to approximating the edit distance between $A$ and $B$.
                
                Without loss of generality assume $A$ is $\rho$-balanced.
                This means that all of its character frequencies are at least $(1/s - \rho)n$.
                So there exists a unary common subsequence of $A$ and $B$ of at least this length.
                If this is a $(1/s+\gamma)$ approximation we are done.
                Otherwise the LCS must be quite large:
                
                    \[\lcs(A,B) > s(1/s - \rho)n - s\gamma n = (1 - s(\rho + \gamma))n.\]
                    
                Recall from \Cref{lm:lcs-ed} that 
                    \[\ed(A,B) + 2\cdot \lcs(A,B) = 2n,\]
                where $\ed(A,B)$ denotes the (no substitutions) edit distance of $A$ and $B$.
                Using the ED approximation of \cite{simpler-ED} with some approximation ratio $c > 3$ we recover an LCS approximation of length at least
                    \[n - c\grp{n - \lcs(A,B)} > n\grp{1 - cs(\rho+\gamma)}\]
                so as long as we have $1 - cs(\rho+\gamma) \ge 1/s + \gamma$,
                this approximation is strong enough.
                This inequality holds when
                    \[\gamma \le \frac{s - 1 -cs^2\rho}{s(1+cs)}.\]
                    
                We can ensure it does by picking $\rho$ small enough that the numerator of the right hand side above is positive and then taking $\gamma$ smaller than the right hand side.
                Note that smaller values of $\rho$ correspond to larger values of $\gamma$.
            \end{proof}
        
        \begin{lemma}[Balanced LCS]
            \label{lem:balanced-lcs}
            If the LCS of $A$ and $B$ is \itnot $\rho$-balanced, then in linear time we can $(1/s + \rho/(s-1))$-approximate $\lcs(A,B)$.
        \end{lemma}
            \begin{proof}  
        Returning the longest common unary subsequence gives the desired approximation.
        To see this, let $\sigma_{\max}$ and $\sigma_{\min}$ 
        be the most frequent and least frequent characters in the LCS respectively.
        Since the LCS is not balanced, either
        $\sigma_{\max}$ makes up more than a $(1/s + \rho)$ fraction of all symbols in the LCS,
        or $\sigma_{\min}$ makes up fewer than a $(1/s - \rho)$ fraction of the symbols in the LCS.
        
        In the latter case, the $s-1$ members of the alphabet besides $\sigma_{\min}$ must account for at least an $((s-1)/s + \rho)$ fraction of characters in the LCS.
        Among these, $\sigma_{\max}$ appears the most often, which means by averaging that $\sigma_{\max}$ accounts for at least a \[\frac{1}{s} + \frac{\rho}{s-1}\]
        fraction of all symbols in the LCS.
        
        In either case, $\sigma_{\max}$ makes up at least a $(1/s + \rho/(s-1))$ fraction of the symbols in the LCS.
        Then the string consisting of $\sigma_{\max}$ repeated $\min(\sigma_{\max}(A), \sigma_{\max}(B))$ times is a common subsequence of $A$ and $B$ which has at least as many instances of $\sigma_{\max}$ as the LCS does, and thus yields the desired approximation.
            \end{proof}
            
        After handling these easy cases, our approach is to restrict $A$ and $B$ to binary strings and invoke the frequency arguments from previous work.
        As we noted before, we cannot directly use the reduction in \Cref{corr:bin-reduction} because the better than $1/2$ approximation of \cite{approxED-LCS} only applies when the inputs have the same length.
        It turns out however, that the arguments of \cite{approxED-LCS} \emph{do} hold for strings of differing length as long as the inputs satisfy a nice frequency condition.
        The following lemma demonstrates how by careful choice of subalphabets we can find restrictions meeting this condition.
        To describe this condition, we introduce some new notation:
        given a string $A$ from an alphabet $\Sigma$, for any symbol $\sigma\in \Sigma$ we let $\sigma(A)$ denote the number of times $\sigma$ appears in $A$.
    
    \begin{lemma}[Binary Restriction]
        \label{lem:weak-restrict}
        Suppose neither $A$ nor $B$ are $\rho$-balanced.
        Then there exists $\Sigma'\sub\Sigma$ with $|\Sigma'| = 2$
        such that the restrictions of $A$ and $B$ to $\Sigma'$ are not $\rho/s$-balanced.
    \end{lemma}
    
    \begin{proof}
        
        Let $\alpha_1 \le \dots \le \alpha_s$ be the number of times the distinct symbols $\sigma_1, \dots, \sigma_s$ of $\Sigma$ appear respectively in $A$,
        so that $\sigma_s$ is the most frequent symbol of $A$ and $\sigma_1$ is the least common character in $A$.
        By averaging we know that $\alpha_s \ge n/s \ge \alpha_1$.
        Since $A$ is not $\rho$-balanced we deduce that
            \[\alpha_s - \alpha_1 > \rho n.\]
            
    We can decompose the left hand side of the above equation to
            \[\alpha_s - \alpha_1 = \sum_{i=2}^{s} \grp{\alpha_i - \alpha_{i-1}}.\]
        Since all the summands on the right hand side are positive, there exists some index $j$ with
        \begin{equation}
        \label{eq:apart}
        \alpha_j - \alpha_{j-1} > (\rho/s)n.
        \end{equation}
        Note that we can find such a $j$ in linear time by scanning through $A$ and $B$ and keeping counts of all the characters that appear.
        Now, consider the $(s-1)$ two-element sets  
            \begin{equation}
            \label{eq:sets}
            \set{\sigma_s, \sigma_{j-1}}, \set{\sigma_{s-1}, \sigma_{j-1}}, \dots, \set{\sigma_j, \sigma_{j-1}},
              \set{\sigma_{j}, \sigma_{j-2}}, \dots, \set{\sigma_{j}, \sigma_{1}}.
            \end{equation}
            
        We claim that one of these sets satisfies the properties of $\Sigma'$ from the lemma statement.
        First, note that \Cref{eq:apart} ensures that the restriction of $A$ to any of the above sets is not $\rho/s$-balanced,
        so it suffices to verify that the restriction of $B$ will not be balanced for one of these sets.
        Suppose to the contrary that none of the sets from \Cref{eq:sets} satisfies the desired conditions.
        We will use a triangle-inequality argument on the character frequencies of $B$ to derive a contradiction.
        Let $M$ and $m$ be indices such that $\sigma_M$ is the most frequent character of $B$ and $\sigma_m$ is the least frequent.
        Since $B$ is not $\rho$-balanced, we know that $M\neq m$.
        If $M \ge j > m$ we may write
        
            \[\abs{\sigma_M(B) - \sigma_m(B)} \le \abs{\sigma_M(B) - \sigma_{j-1}(B)} + \abs{\sigma_{j-1}(B) - \sigma_j(B)}+  \abs{\sigma_{j}(B) - \sigma_m(B)}.\]
            
        By assumption, the restrictions of $B$ to the sets $\set{\sigma_M, \sigma_{j-1}}$, $\set{\sigma_j, \sigma_{j-1}}$, and $\set{\sigma_{j}, \sigma_m}$ are $(\rho/s)$-balanced.
        Thus, each addend on the right hand side of the above inequality is bounded above by $(\rho/s)n$.
        It follows that
            \[\abs{\sigma_M(B) - \sigma_m(B)} \le (3\rho/s)n.\]
        By similar reasoning, if $m = j$ and $M\ge j$ then we have
            \[\abs{\sigma_M(B) - \sigma_m(B)} = \abs{\sigma_M(B) - \sigma_j(B)} \le \abs{\sigma_M(B) - \sigma_{j-1}(B)} + \abs{\sigma_{j-1}(B) - \sigma_j(B)}  \le (2\rho/s)n.\]

        If instead $M, m > j$ we know that
            \[\abs{\sigma_M(B) - \sigma_m(B)} \le \abs{\sigma_M(B) - \sigma_{j-1}(B)} + \abs{\sigma_{j-1}(B) - \sigma_{m}(B)} \le (2\rho/s)n\]
        since now the subalphabets $\set{\sigma_M, \sigma_{j-1}}$ and $\set{\sigma_m, \sigma_{j-1}}$ both occur in \Cref{eq:sets}.
        Similar reasoning on the remaining cases of the values of $M$ and $m$ relative to $j$ to establishes the inequality
            \[\sigma_M(B) - \sigma_m(B) \le (3\rho/s)n.\]
        We have dropped absolute value signs because the left hand side of the above equation is positive by definition of $M$ and $m$.
        Since $s\ge 3$ this contradicts the assumption that $B$ is not $\rho$-balanced, and the desired result follows.
        Note that we can find which of subalphabets from \Cref{eq:sets} satisfies the conditions of the lemma in $O(n + s)$ time by scanning through $B$ to get the counts of each of its characters and trying out all the restrictions.
    \end{proof}
    
    In our proof of \Cref{thm:equal-length-intro}, we will require a stronger form of the result from Section 4 of \cite{approxED-LCS}.
This variant of their theorem is useful because it applies to strings of different length.

\begin{lemma}
\label{lm:imbalanced}
Let $X$ and $Y$ be binary strings of length $n$ and $m$ respectively, where $m\le n$.
Suppose the frequencies $0(Y)$ and $1(X)$ are both at most $(1/2 - \rho)m$ for some positive constant $\rho$.
Then there exist positive constants $\delta$ and $\epsilon$ such that if $0(Y)$ and $1(X)$ are within $\delta m$ of each other, we can compute a $\grp{1/2+\epsilon}$ approximation of $\lcs(X,Y)$ in subquadratic time.

\end{lemma}

Although the realization that this type of result holds for strings of differing length is novel,
the proof of \Cref{lm:imbalanced} itself is conceptually identical to the frequency analysis used in \cite{approxED-LCS}, requiring only minor changes to make the argument go through.
For completeness we include the detailed casework proof of this result in \Cref{sec:not-an-appendix}.

Finally we apply a lemma that provides some frequency information about the strings.
    
    \begin{lemma}[Lemma 3.1 from \cite{approxED-LCS}]
        \label{lem:distant-freq}
        For any $\delta > 0$
        if
            \[\min\grp{0(X), 0(Y)} > (1+\delta)\min\grp{1(X), 1(Y)}\]
        or	
	        \[\min\grp{1(X), 1(Y)} > (1+\delta)\min\grp{0(X), 0(Y)}\]
	        
	    then there is a unary $(1+\delta)/(2+\delta)$ approximation for $\lcs(X,Y)$.
    \end{lemma}
        
        \begin{proof}
            Observe that for any binary strings $X$ and $Y$ we have
                \begin{equation}
                \label{eq:freq-bound}
                \lcs(X,Y) \le \min\grp{0(X), 0(Y)} + \min\grp{1(X), 1(Y)}.
                \end{equation}
            This equation holds because the LCS is a subsequence of $X$ and $Y$, and thus cannot contain more zeros or ones than either of the strings $X$ or $Y$ individually have.
            
            Suppose by symmetry that $\min\grp{0(X), 0(Y)}$ is the larger of the two addends on the right.
            Then the we can return the all zeros string of this length.
            By the first inequality we get
                \[\min\grp{0(X), 0(Y)} > (1+\delta)\grp{\lcs(X,Y) - \min\grp{0(X), 0(Y)}}\]
            and then rearranging proves the claim.
        \end{proof}
        
    We now combine these results to improve LCS approximation on all alphabets.

        \begin{proof}[Proof of \Cref{thm:equal-length}]
        
            Let $\rho$ and $\rho'$ be positive parameters whose values will be specified later.
            If either of the strings $A$ or $B$ are $\rho s$-balanced, we are done by \Cref{lm:balanced-string} (by taking $\rho$ to be small enough so that the lemma applies).
            If the LCS of $A$ and $B$ is not $\rho'$-balanced we are done by \Cref{lem:balanced-lcs}.
            So, we may assume that neither $A$ nor $B$ are $\rho s$-balanced and that their LCS is $\rho '$-balanced.
            
            By \Cref{lem:weak-restrict} we can find binary alphabet restrictions $X$ and $Y$ of $A$ and $B$ respectively with the property that neither $X$ nor $Y$ are $\rho$-balanced.
            Informally, since the LCS of $A$ and $B$ is balanced, 
            a better than $1/2$ approximation for the LCS of $X$ and $Y$ acts as a better than $1/s$ approximation for $\lcs(A,B)$.
            More precisely, since $\lcs(A,B)$ is $\rho'$-balanced, each of its characters occurs at least $(1/s - \rho')\lcs(A,B)$ times in the LCS.
            By construction, we also know that $\lcs(X,Y)$ is at least as large as a binary alphabet restriction of $\lcs(A,B)$.
            Thus, it follows that given a positive constant $\epsilon'$, a $1/2 + \epsilon'$ approximation of $\lcs(X,Y)$ has length at least
            
                \[\grp{1/2 + \epsilon'} \cdot 2\grp{1/s - \rho'}\lcs(A,B) = \grp{1/s + 2\epsilon'/s - \rho'  - 2\epsilon' \rho' }\lcs(A,B).\]
                
            By setting $\epsilon = \epsilon'/s$ (for example) and $\rho'$ sufficiently small in terms of $\epsilon'$, the above calculation shows that a $1/2 + \epsilon'$ approximation for the LCS of $X$ and $Y$ acts as a $1/s + \epsilon$ approximation for $\lcs(A,B)$.
            Hence to complete the proof, it suffices to get a better than $1/2$ approximation for $\lcs(X,Y)$.
            We may assume that
                \begin{equation}
                \label{eq:0-min}
                \min\grp{0(X), 0(Y)} \le (1+\delta)\min\grp{1(X), 1(Y)}
                \end{equation}
            and
                \begin{equation}
                \label{eq:1-min}
                \min\grp{1(X), 1(Y)}\le (1+\delta)\min\grp{0(X), 0(Y)}
                \end{equation}
            for some constant $\delta$ since
            otherwise \Cref{lem:distant-freq} yields a better than $1/2$ approximation.
            
            As mentioned previously, \cite{approxED-LCS} gives a better than $1/2$ approximation for the LCS when $X$ and $Y$ have equal length.
            If $X$ and $Y$ have different lengths,
            without loss of generality we assume that $|X| > |Y|$ 
            and $0(Y)\le 1(Y)$.
            We do casework on frequencies in $X$,
            relative to the frequencies in $Y$.
            
Since $X$ is longer than $Y$,
we cannot have both $1(X)\le 1(Y)$ and $0(X)\le 0(Y)$.
            
            If $1(X) > 1(Y)$ and $0(X) > 0(Y)$ simultaneously, then \cref{eq:1-min} implies that
                \[1(Y)\le (1+\delta)\cdot 0(Y)\]
            which contradicts the fact that $Y$ is not $\rho$-balanced as long as we take $\delta \le 2\rho$.
            
            If  instead $1(X)\le 1(Y)$ and $0(X)\le 0(Y)$, by \cref{eq:1-min} we similarly have
                \[0(X)\le 0(Y) \le 1(Y)\le (1+\delta)\cdot 0(X)\]
            which again contradicts the fact that $Y$ is not $\rho$-balanced for the same choice of $\delta$.
            
    Thus, the only possibility is that $1(X)\le 1(Y)$ and $0(X) > 0(Y)$.
    Let $m=|Y|$ be the length of string $Y$.
    Then \cref{eq:0-min} implies that 
        \[0(Y)\le (1+\delta)\cdot 1(X) \le 1(X) + \delta m\]
    while \cref{eq:1-min} implies that
        \[1(X)\le (1+\delta)\cdot 0(Y) \le 0(Y) + \delta m.\]
    
    Consequently, $0(Y)$ and $1(X)$ are within $\delta m$ of each other. 
    Then by \Cref{lm:imbalanced}, as long as we take $\delta$ small enough in terms of $\rho$ we get a subquadratic $1/2 + \epsilon'$ approximation for $\lcs(X,Y)$.
            As noted earlier, this then yields the desired $1/s+\eps$ approximation for $\lcs(A,B)$ in truly subquadratic time.
    In fact, because we only ever use subroutines that run in linear time or constant factor approximations to edit distance which take near-linear time, the overall algorithm takes near-linear time.
        \end{proof}

\section{Proof of \Cref{lm:imbalanced}}
    \label{sec:not-an-appendix}
    This section is devoted to proving \Cref{lm:imbalanced}.
    We do this working through the individual arguments in the case analysis of \cite{approxED-LCS} and verifying that the arguments still hold in the case where the strings have different lengths, as long as they satisfy the frequency requirements included in the hypotheses of the lemma.
    Throughout this section we fix the binary alphabet $\Sigma = \set{0,1}$ and assume that all strings come from this alphabet.
    
    We carry over the following subroutines from \cite{approxED-LCS}.
    
    \begin{definition}
    \label{def:match}
    Given strings $A$ and $B$ and a symbol $\sigma$, the algorithm $\match(A,B,\sigma)$ returns the largest subsequence of $A$ and $B$ consisting entirely of copies of $\sigma$.
    \end{definition} 
    
    \begin{definition}
    \label{def:bm}
    Given strings $A$ and $B$, the algorithm $\bm(A,B)$ returns the longer of the strings
        \[\match(A,B,0) \quad \text{and} \quad \match(A,B,1).\]
        
    In other words, the algorithm returns the largest common unary subsequence.
    \end{definition} 
    
    Note that $\bm$ is a $1/2$ approximation algorithm for LCS.
    
    \begin{definition}
    \label{def:greed}
    Given strings $A_1, A_2$, and $B$, the algorithm $\greed(A_1, A_2, B)$ returns 
        \[\max_{B = B_1\sqcup B_2} \bm(A_1,B_1) + \bm(A_2,B_2)\]
    taken over all possible splits of $B$ into two contiguous left and right substrings $B_1$ and $B_2$.
    \end{definition}
    
    These algorithms can all be implemented to run in linear time by scanning through the counting how many $0$s and $1$s appear in each string.
    The final procedure we will invoke utilizes an approximation algorithm for computing the edit distance of two strings as a black-box.
    Although any fast enough constant factor approximation for edit distance will work, for concreteness we assume that it leverages the algorithm $\widetilde{\ed}$ from \cite{simpler-ED}, which runs in subquadratic time and can approximate the edit distance to a $c = 3 + \epsilon'$ factor for any fixed positive constant $\epsilon'$.
    Note that this algorithm can also return a common subsequence achieving the given length.
    
    \begin{proof}[Proof of \Cref{lm:imbalanced}]
    
    Let $1(X) = \alpha m$ for some $\alpha\in [0,1]$.
    By assumption, we know that $\alpha < 1/2$ is bounded away from $1/2$ by some constant amount.
    We can also assume that $0(Y)$ is within $\delta m$ of $1(X)$ for some positive parameter $\delta < 0$ to be picked later on.
    We will end up setting $\delta$ to be some sufficiently small constant depending on $\alpha$.
    We will use the notation $\approx$ to denote quantities that are within $\delta m$ of each other.
    For example $1(X)\approx 0(Y)$.
    This lets us avoid having to stick in $\pm \delta m$ symbols in all the inequalities and helps make the arguments cleaner without affecting the correctness (since we just care about getting a $1/2 + \epsilon$ approximation for \emph{some} constant $\epsilon > 0$).
    
    Note that by \Cref{eq:freq-bound} the LCS of the input strings
    
        \begin{equation}
        \label{eq:triv-bound}
        \lcs(X,Y) \le (2\alpha + \delta)m
        \end{equation}
    cannot be too large, because the LCS contain at most $\alpha m$ $1$s from $X$ and $\approx \alpha m$ $0$s from $Y$.

    Let $R_X$ and $R_Y$ denote the substrings of $X$ and $Y$ consisting of their rightmost $\alpha m$ characters respectively.
    Similarly define $L_X$ and $L_Y$ as the substrings of $X$ and $Y$ consisting of the leftmost $\alpha m$ characters.
    Set $M_X = X\setminus (L_X\cup R_X)$ 
    and $M_Y = Y\setminus (L_Y\cup R_Y)$
    to be the middle substrings of $X$ and $Y$ that remain when these left and right ends are chopped off.
    
    We now follow the casework of \cite{approxED-LCS}, explaining at each step why the arguments still hold in our more general setting.
    These cases are based off the frequencies of $0$s and $1$s in the left and right ends of the inputs, and consider separately the situation where these substrings are pseudorandom (balanced) or structured (imbalanced).
    In the former case we can appeal to edit distance as in the proof of \Cref{lm:balanced-string}, and in the latter situation we can exploit the imbalance in the strings to use the simpler unary algorithms described in definitions \ref{def:match}, \ref{def:bm}, and \ref{def:greed}.
    Intuitively, we succeed in using edit distance approximation arguments (and overcome the barrier described the end of \Cref{sec:input-length}) in this particular case because even though $X$ and $Y$ may have different size, we identify right and left substrings which all have equal length and  only employ edit distance approximation around these areas.
    
    Recall that we assumed that $\alpha < 1/2 - \rho$ for some constant $\rho$.
    Take a parameter $\beta < \rho/20$ to represent deviation from the balanced case.
    By choosing $\delta$ sufficiently small in terms of $\beta$ we may additionally assume that
    
        \begin{equation}
            \label{eq:imbalanced}
            1(X), 0(Y) < \grp{1/2 - 10\beta}m
        \end{equation}
    since in the hypothesis of the lemma we supposed that $1(X) < (1/2 - \rho)m$ and that $0(Y) \approx 1(X)$.
    Finally, as one last piece of notation, we write $|A|$ to denote length of an arbitrary string $A$.
    We now begin the casework, maintaining consistency with \cite{approxED-LCS}.
            
    \paragraph*{Case 1(a): $1(R_Y), 0(R_X) \in \left[\grp{\alpha/2 - 4\beta}m, \grp{\alpha/2 + 4\beta}m\right]$}
            In this case both right ends of the strings are balanced.
            We will give a good approximation for the LCS by splitting these strings at the right ends and using the aforementioned algorithms.
            
        Consider an optimal alignment between $X$ and $Y$ (i.e. a maximum partial matching of identical characters in $X$ and $Y$, corresponding to the LCS of $X$ and $Y$).
        Let $\widehat{R_Y}$ be the minimal suffix (from the right end) of the string $Y$ with the property that every character from $R_X$ which is matched in the alignment is paired up with some character in $\widehat{R_Y}$.
        
        Without loss of generality, we may assume that $\widehat{R_Y}$ is a substring of $R_Y$.
        This is because if $\widehat{R_Y}$ was not contained in $R_Y$, we could define an analogous substring of $\widehat{R_X}$ of $X$ satisfying $\widehat{R_X}\sub R_X$ and then use a symmetric argument to get the desired approximation.
        Let $\widehat{L_Y} = Y\setminus \widehat{R_Y}$
        be the left substring of $Y$ that remains after chopping of $\widehat{R_Y}$.
        
        Optimality of the alignment implies that
        
            \begin{equation}
            \label{eq:1bi-split}
            \lcs(X,Y) = \lcs(X\setminus R_X, \widehat{L_Y}) + \lcs(R_X, \widehat{R_Y}).
            \end{equation}
            
        Define the quantities
            \[f_L = \min(1(X\setminus R_X), 1(\widehat{L_Y})) + 
            \min(0(X\setminus R_X), 0(\widehat{L_Y}))\]
        and 
            \[f_R = \min(1(R_X), 1(\widehat{R_Y})) + \min(0(R_X), 0(\widehat{R_Y}))\]
            
        which represent frequency-based upper bounds for the LCS terms from the right hand side of \Cref{eq:1bi-split}.
        They are useful because \Cref{eq:freq-bound} together with \Cref{eq:1bi-split} implies that
        
            \begin{equation}
            \label{eq:abbrv-bound}
            \lcs(X,Y) \le f_L + f_R.
            \end{equation}
        
        We also introduce the quantity
        
            \[Z = \max\grp{\min(1(X\setminus R_X), 1(\widehat{L_Y})), \min(0(X\setminus R_X), 0(\widehat{L_Y}))}\]
            
        which is the larger of the two addends defining $f_L$, and equal to the length of the string returned by
        
        \[\bm(X\setminus R_X, \widehat{L_Y}).\]
        
        We further subdivide into cases based off the size of $Z$.
                
        \textbf{Case 1(a)(i): $Z > \grp{\alpha/2 + 10\beta}m$}
        
        When $Z$ is large we can combine two unary subsequences to get a good enough LCS approximation.
        We first show that $Z$ is bigger than $f_L/2$ by a constant fraction of $m$.
        
        By definition we have
    
            \[f_L - Z = \min\grp{\min(1(X\setminus R_X), 1(\widehat{L_Y})), \min(0(X\setminus R_X), 0(\widehat{L_Y}))} \le 1(X\setminus R_X).\]
            
        Since $X$ has $\alpha m$ ones and $R_X$ has length $\alpha m$ we get that
        
            \[1(X\setminus R_X) = \alpha m - 1(R_X) = \alpha m - \grp{\alpha m - 0(R_X)} = 0(R_X).\]
            
        Then using the case assumptions we have
        
            \[0(R_X) \le \grp{\alpha/2 + 4\beta}m < Z - 6\beta m. \]
            
        Chaining these inequalities together and rearranging we deduce that
        
            \[Z > f_L/2 + 3\beta m.\]
            
        Now if we make a single call to the $\greed$ routine we obtain a string of length
        
            \begin{equation}
            \label{eq:greed}
            \greed(X\setminus R_X, R_X, Y) \ge 
            \bm(X\setminus R_X, \widehat{L_Y}) + \bm(R_X, \widehat{R_Y}).
            \end{equation}
            
        From our earlier discussion we have 
        
        \[\bm(X\setminus R_X, \widehat{L_Y}) \ge Z > f_L/2 + 3\beta m.\]
        
        Moreover
        
        \[\bm(R_X, \widehat{R_Y}) = \max( \min(1(R_X), 1(\widehat{R_Y})), \min(0(R_X), 0(\widehat{R_Y}))) \ge f_R/2\]
        
        since we are taking the maximum over two addends that sum to $f_R$. 
        
        Finally, if we substitute the above inequalities into \Cref{eq:greed} and apply \Cref{eq:abbrv-bound} we get that
        
            \[\greed(X\setminus R_X, R_X, Y) \ge \grp{f_L + f_R}/2 + 3\beta m \ge \lcs(X,Y)/2 + 3\beta m \ge \grp{1/2 + 3\beta}\lcs(X,Y).\]
            
        Thus running $\greed$ gives us a better than $1/2$ approximation in this case.
        
        \textbf{Case 1(a)(ii) : $Z \le \grp{\alpha/2 + 10\beta}m$}
        
Intuitively, in this case $Z$ is too small for us to ensure a good approximation using frequency guarantees alone.
However, because $Z$ is so small, the LCS also cannot be too large.
Because of this, we will be able to get a good approximation by combining a common subsequence of the (balanced) right ends of the strings with a common subsequence of the strings with the right ends removed.

More concretely, we leverage the $\approxed$ algorithm from \Cref{def:approxed}.
        From our previous observation we have
        $f_L\le 2Z \le \grp{\alpha + 20\beta}m.$
        So via \Cref{eq:1bi-split} we can bound the LCS by
            \begin{equation}
            \label{eq:lcs-upper}
            \lcs(X,Y) \le f_L + \lcs(R_X, \widehat{R_Y}) \le \grp{\alpha + 20\beta}m + \lcs(R_X, R_Y)
            \end{equation}
        where in the last step we also used the fact that $\widehat{R_Y}\subseteq R_Y$.
        
        We will get an approximation by returning a unary subsequence from the left parts of the strings, and using edit distance approximation on the right ends.
        First, we can make a call to $\bm$ and get a common subsequence of length at least
            \[\bm(X\setminus R_X, Y\setminus R_Y) \ge \match(X\setminus R_X, Y\setminus R_Y, 0) \ge \grp{\alpha/2 - 4\beta}m.\]
        This last inequality above follows from combining the case 1 assumptions about the frequencies of characters in $0(R_X)$ and $1(R_Y)$,
        together with the facts that $0(Y)\approx 1(X) = \alpha m$ and $|R_X|=|R_Y| = \alpha m$.
        
        Now, since  $R_X$ and $R_Y$ are $(4\beta/\alpha)$-balanced we can apply \Cref{lm:balanced-string} with $n=\alpha m$ as long as we take $\beta$ sufficiently small in terms of $\alpha$.
        Note that for fixed alphabet size and $\rho$, the parameter $\beta$ remains $\Omega(1)$.
        This ensures that in subquadratic time we can compute a common subsequence of $R_X$ and $R_Y$ with length at least
        
            \[\grp{1/2 + \gamma} \lcs(R_X, R_Y) \]
        where $\gamma < 1$ is the constant from \Cref{lm:balanced-string}, which is larger for smaller values of $\beta$.
        By the frequency assumptions on $R_X$ and $R_Y$, we know that
        
            \[\lcs(R_X, R_Y) \ge \grp{\alpha/2 - 4\beta}m\]
        so in fact the subsequence of $R_X$ and $R_Y$ returned by \Cref{lm:balanced-string} has length at least
        
            \[\grp{1/2 + \gamma} \lcs(R_X, R_Y)  \ge \lcs(R_X,R_Y)/2 + \gamma(\alpha/2 - 4\beta)m.\]
            
        Now if we combine these two subsequences together, we get a common subsequence of $X$ and $Y$ of length at least
            \[\lcs(R_X, R_Y)/2 + \grp{\alpha/2 - 4\beta}m + \gamma(\alpha/2 - 4\beta)m\]
        which can be written in the form
            \[\grp{(\alpha + 20\beta)m + \lcs(R_X, R_Y)}/2 + 
            \grp{\gamma\alpha/2 - 14\beta - 4\gamma\beta}m.\]
            
        By applying \Cref{eq:lcs-upper} and the fact that $\gamma < 1$ we see that this is at least
        
            \[\lcs(X,Y)/2 + \grp{\gamma\alpha/2 - 18\beta }m.\]
            
        Finally, by picking $\beta$ small enough this expression is at least
        
        \[\lcs(X,Y)/2 + \grp{\gamma\alpha/3}m \ge \grp{1/2 + \gamma/6}\lcs(X,Y)\]
        where in the last step we have used \Cref{eq:freq-bound} together with $1(X) = \alpha m$ and $0(Y) \approx \alpha m$.
        This proves that we can attain a better than $1/2$ approximation for the LCS as claimed.
    
    \paragraph*{Case 1(b): $1(R_Y) < \grp{\alpha/2 - 4\beta}m$ and $0(R_X) \le \grp{\alpha/2 + 2\beta}m$}
    
        In this case the right ends of $X$ and $Y$ each do not contain too much their respective strings' most common characters.
        We show that this implies the LCS of both strings must be so small that simply returning a unary string yields a better than $1/2$ approximation.
        
        As in case 1(a), consider an optimal alignment between $X$ and $Y$.
        Now define the substring $\widehat{R_Y}$ to be the minimal suffix of $Y$ which contains all characters of $Y$ that $R_X$ is aligned to.
        Let $\widehat{L_Y} = Y\setminus \widehat{R_Y}$ be what remains of $Y$ after $R_Y$ is removed.
        Since the alignment is optimal, \Cref{eq:1bi-split} holds.
            
We now subdivide into further case based off how the right ends of strings are aligned.

        \textbf{Case 1(b)(i): every character of $R_X$ is matched to some character of $R_Y$ in the alignment.}
        
        In this case $\widehat{R_Y}$ is a substring of $R_Y$.
        From \Cref{eq:freq-bound} and the case assumptions we get that
        
            \[\lcs(X\setminus R_X, \widehat{L_Y})\le 1(X\setminus R_X) + 0(\widehat{L_Y}) \le \grp{\alpha/2 + 2\beta}m + 0(\widehat{L_Y})\]
            
        and
        
            \[\lcs(R_X, \widehat{R_Y}) \le 1(\widehat{R_Y}) + 0(\widehat{R_Y}) \le \grp{\alpha/2 - 4\beta}m + 0(\widehat{R_Y}).\]
            
        Note that in this second inequality we are using the inequality $1(\widehat{R_Y})\le 1(R_Y)$ which follows from the earlier observation that $\widehat{R_Y} \sub R_Y$.
        By adding these inequalities together and substituting the result into \Cref{eq:1bi-split} we deduce that
        
            \[\lcs(X,Y)\le \grp{\alpha - 2\beta}m + \grp{0(\widehat{L_Y})+ 0(\widehat{R_Y})} \le \grp{\alpha - 2\beta}m + \alpha m = \grp{2\alpha - 2\beta}m\]
            
        where we have used the fact that $0$ occurs in $Y$  $\approx \alpha m$ times.
        
        \textbf{Case 1(b)(ii): some character of $R_X$ is matched outside $R_Y$ in the alignment.}
        
        In this case  $R_Y$ is a substring of $\widehat{R_Y}$.
            
        Using \Cref{eq:freq-bound} again we find that
        
            \begin{equation}
            \label{eq:first-term}
            \lcs(X\setminus R_X, \widehat{L_Y}) \le 1(X\setminus R_X) + 0(\widehat{L_Y})
            \end{equation}
            
        From $0(R_X)\le \grp{\alpha/2 + 2\beta}m$ we know that $1(R_X)\ge \grp{\alpha/2 - 2\beta}m$
        since $R_X$ has length $\alpha m$.
        It follows that
        
            \[1(X\setminus R_X) = 1(X) - 1(R_X) \le \grp{\alpha/2 + 2\beta}m\]
            
        since $1(X) = \alpha m$.
        
        Similarly, since $1(R_Y) < \grp{\alpha/2 - 4\beta}m$ and $R_Y$ has length $\alpha m$ we know that $0(R_Y) > \grp{\alpha/2 + 4\beta}m.$
        Since $R_Y\sub \widehat{R_Y}$ it must be the case that $0(\widehat{R_Y}) > \grp{\alpha/2+4\beta}m$
        which means that
        
            \[0(\widehat{L_Y}) = 0(Y) - 0(\widehat{R_Y}) \approx \alpha m - 0(\widehat{R_Y}) < \grp{\alpha/2 - 4\beta}m.\]
            
        Adding these two inequalities and substituting into \Cref{eq:first-term} proves that
        
            \[\lcs(X\setminus R_X, \widehat{L_Y}) \le \grp{\alpha - 2\beta}m.\]
            
        Then applying \Cref{eq:1bi-split} and using the fact that $R_X$ has length $\alpha m$ proves that
        
            \[\lcs(X,Y) = \lcs(X\setminus R_X, \widehat{L_Y}) + \lcs(R_X, \widehat{R_Y}) \le (\alpha - 2\beta)m + \alpha m = \grp{2\alpha - 2\beta}m.\]
                
        Thus in both this subcase and the previous one, the LCS is at most $\grp{2\alpha - 2\beta}m$.
        Hence, returning the string of $\approx \alpha m$ zeros obtained by calling $\match(X, Y, 0)$ gives a better than $1/2$ approximation as desired.
    
    \paragraph*{Case 1(c): $1(R_Y) \le \grp{\alpha/2 + 2\beta}m$ and $0(R_X) < \grp{\alpha/2 - 4\beta}m$}
    
    This case is symmetric to case 1(b) and similar reasoning handles it.
            
    \paragraph*{Case 2: $1(L_Y), 0(L_X)\le \grp{\alpha/2 + 2\beta}m$}
    
    Combining cases 1(a), 1(b), and 1(c) resolves the situation where 
    
    \[1(R_Y), 0(R_X)\le \alpha/2 + 2\beta.\]
    
    Consequently, case 2 is symmetric to case 1.
    In particular, we can flip the strings (by replacing $X$ with the string $X'$ which consists of the symbols of $X$ read in reverse from right to left, and replacing $Y$ with its analogous reverse string $Y'$) and apply the arguments from case 1 to handle this case.
    
    \paragraph*{Case 3: $1(L_Y), 1(R_Y)\le \grp{\alpha/2 + \beta}m$ and $0(L_X), 0(R_X) > \grp{\alpha/2+2\beta}m$}
    
    In this case both ends of the inputs have many $0$s,
    so repeated calls to $\match$ will be enough to guarantee a better than $1/2$ approximation.
    
    Since $L_Y$ has length $\alpha m$ and at most $\grp{\alpha/2+\beta}m$ instances of $1$, it must have at least
        \[0(L_Y) \ge \alpha m - (\alpha/2 + \beta)m = (\alpha/2 - \beta)m\]
    
    occurrences of $0$.
    Since $1(R_Y)\le \grp{\alpha/2 + \beta}m$ the same reasoning shows that
    \[0(R_Y) \ge \grp{\alpha/2 - \beta}m.\]
    Consequently we can get common subsequences of the left and right ends of the strings consisting entirely of $0$s with lengths
    
        \[\match(L_X, L_Y, 0) = \min(0(L_X), 0(L_Y)) \ge \grp{\alpha/2 - \beta}m\]
       
    and
            \[\match(R_X, R_Y, 0) = \min(0(R_X), 0(R_Y)) \ge \grp{\alpha/2 - \beta}m.\]
            
    Since the ends have many $0$s, we expect the middle substrings to have many $1$s.
    Indeed, since the left end of $X$ 
    
        \[1(L_X) = |L_X| - 0(L_X) < \alpha m - (\alpha/2 + 2\beta )m = \grp{\alpha/2 - 2\beta}m\]
        
    does not have many $1$s and similar reasoning shows that
        \[1(R_X) < \grp{\alpha/2 - 2\beta}m\]
    this holds for the right end as well, the middle of $X$ has only
        \[1(M_X) = 1(X) - 1(L_X) - 1(R_X) > \alpha m - 2(\alpha/2 - 2\beta)m = 4\beta m\]
    appearances of $1$.
    
    To argue that $M_Y$ has enough $1$s we will need to appeal to the fact that $Y$ is not balanced.
    From \Cref{eq:imbalanced} we know that
    
        \[1(Y) = m - 0(Y) > m - (1/2 - 10\beta)m = \grp{1/2 + 10\beta}m.\]
        
    We can use the case assumptions together $\alpha < 1/2$ to then deduce
    
        \[1(M_Y) = 1(Y) - 1(L_Y) - 1(R_Y) > (\alpha + 10\beta)m - 2(\alpha/2 + \beta)m = 8\beta m.\]
    Consequently the largest all $1$s subsequence between the middle substrings has length at least
    
        \[\match(M_X, M_Y, 1) = \min(1(M_X), 1(M_Y)) \ge 4\beta m.\]
        
    Thus, by combining these three subsequences from three calls to match we can recover in linear time a common subsequence of size
    
        \[\match(L_X, L_Y, 0) + \match(M_X, M_Y, 1) + \match(R_X, R_Y, 0) > 2(\alpha/2 - \beta)m + 4\beta m = (\alpha + 2\beta)m.\]
        
    By \Cref{eq:triv-bound} the true LCS has length at most $\approx 2\alpha m$, so a string of length $\grp{\alpha + 2\beta}m$ is a $1/2 + \epsilon$ approximation for some constant $\epsilon < 2\beta/\alpha$ as desired.

    \paragraph*{Case 4: $0(L_X), 0(R_X)\le \grp{\alpha/2 + \beta}m$ and $1(L_Y), 1(R_Y) > \grp{\alpha/2 + 2\beta}m$}
    
    This case is symmetric to case 3 and similar reasoning handles it.

    \paragraph*{Case 5: $0(L_X), 1(R_Y) > \grp{\alpha/2 + \beta}m$}
    In this case, the ends of the strings have unusually large instances of either $0$ or $1$.
    This will enable us to combine two calls to $\match$ to get the desired approximation.
        
    We can check that the right end of $Y$ does not have many $0$s
    
    \[0(R_Y) = |R_Y| - 1(R_Y) < \alpha m - (\alpha/2 + \beta)m = \grp{\alpha/2-\beta}m.\]
    
    It follows that the remainder of the string $Y\setminus R_Y = L_Y\cup M_Y$ has many zeros
    
                \[0(Y\setminus R_Y) = 0(Y) - 0(R_Y) \approx \alpha m - 0(R_Y) > \grp{\alpha/2 + \beta}m.\]
                
            Similar reasoning shows that
            
            \[1(L_X) = \alpha m - 0(L_X) < \grp{\alpha/2-\beta}m\]
            
            so that the remaining portion of $X$ has many ones
            
                \[1(X\setminus L_X) = 1(X) - 1(L_X) = \alpha m - 1(L_X) > \grp{\alpha/2+\beta}m.\]
                
            It follows that 
                \[\match(L_X, Y\setminus R_Y, 0) + \match(X\setminus L_X, R_Y, 1) > \grp{\alpha + 2\beta}m\]

            yields a better than $1/2$ approximation since by \Cref{eq:triv-bound} the LCS is at most $\approx 2\alpha m$.
            
        \paragraph*{Case 6: $1(L_Y), 0(R_X) > \grp{\alpha/2 + \beta}m$}
        
        This is symmetric to case 5 and a similar argument proves the result holds in this situation.

        By inspection or by referring to Table 1 of \cite{approxED-LCS}, we can verify that these cases handle all possible input strings satisfying the conditions of the lemma.
        Since in every case we obtain a better than $1/2$ approximation for the LCS in subquadratic time, we have proven the claim.
        
    \end{proof}

\section{Open Problems}
\Cref{thm:equal-length-intro} of our work shows how to obtain better than $1/|\Sigma|$ approximations for the longest common subsequence of \emph{equal-length strings} over an alphabet $\Sigma$.
It remains an open problem to get such approximations in the setting where the input strings have different length (and by \Cref{corr:bin-reduction}, to solve this problem it suffices to obtain an improvement in the setting of binary alphabets, where $|\Sigma| = 2$).

Additionally, although our algorithm beats the longstanding trivial approximation ratio for LCS, it does so only by a modest amount.
It would be interesting to get better approximation ratios, both in terms of their concrete value for small $|\Sigma|$ and in terms of their growth as a function of the alphabet size.


\bibliography{lcs}

\end{document}